\documentclass[final]{siamart190516}

\usepackage{amsfonts}
\usepackage{graphicx}
\usepackage{epstopdf}
\usepackage{algorithmic}
\ifpdf
  \DeclareGraphicsExtensions{.eps,.pdf,.png,.jpg}
\else
  \DeclareGraphicsExtensions{.eps}
\fi

\newsiamremark{remark}{Remark}
\newsiamremark{hypothesis}{Hypothesis}
\crefname{hypothesis}{Hypothesis}{Hypotheses}
\newsiamthm{claim}{Claim}
\newsiamthm{assumption}{Assumption}
\newsiamthm{example}{Example}

\headers{Utility-based acceptability indices}{M. Pitera and M. R\'asonyi}

\title{Utility-based acceptability indices}

 \author{
Marcin Pitera\thanks{Institute of Mathematics, Jagiellonian University, Krakow, Poland; research supported by NCN grant 2020/37/B/ST1/00463.}
  \and
Mikl\'os R\'asonyi\thanks{Alfr\'ed R\'enyi Institute of Mathematics and E\"otv\"os Lor\'and{}
University, Budapest, Hungary; research supported by the National Research, Development and Innovation Office (NKFIH) grant K 143529.}}

\usepackage{amsopn}
\DeclareMathOperator{\E}{\mathbb{E}}  

\usepackage{enumitem}

\def\cF{\mathcal{F}}

\def\bP{\mathbb{P}}
\def\bQ{\mathbb{Q}}
\def\bE{\mathbb{E}}
\def\bR{\mathbb{R}}
\def\bT{\mathbb{T}}
\def\bN{\mathbb{N}}

\DeclareMathOperator*{\esssup}{ess\,sup} 
\DeclareMathOperator*{\essinf}{ess\,inf} 

\makeatletter
\def\namedlabel#1#2{\begingroup
    #2%
    \def\@currentlabel{#2}%
    \phantomsection\label{#1}\endgroup
}
\makeatother

\begin{document}

\maketitle
\begin{abstract}
In this short paper we introduce a new class of performance measures based on certainty equivalents defined via scaled utility functions. We analyse their properties, show that the corresponding portfolio optimization problem is well-posed under generic conditions, and analyse the link between portfolio dynamics, benchmark process, and utility function choice in the long-run setting. 
\end{abstract}

\section{Introduction}\label{intro}

Accurate measurement of financial position performance and relative benchmarking is one of the cornerstones of the modern portfolio optimization theory and related risk-averse control problems, see \cite{Pri2007,KolTutFab2014} and references therein. The normative goal of performance measurement is to quantify the balance between risk and reward using a single number that establish preference order. Among the most recognised measures of financial performance are the risk-to-reward ratios, with {\it Sharpe Ratio} and {\it Gain-Loss Ratio} being the most common choices, see \cite{Sha1964,OrtBigStoRachFab2005,CheKro2013} for economic background and multiple alternatives. 

In the seminal work \cite{CheMad2009}, the authors propose a set of normative axioms a map should satisfy to be a good indicator of financial performance; this corresponds to properties such as monotonicity and scale invariance. Moreover, the authors of \cite{CheMad2009} provide a dual link between increasing families of coherent risk measures and the family of performance measures, called {\it acceptability indices}, and illustrate how they are interconnected to various financial problems including arbitrage theory and good-deal trade setup. The work presented in \cite{CheMad2009} was followed by multiple other papers which studied {\it acceptability indices} in various context linked to option pricing, good-deal bounds, conic finance, or acceptability maximisation, see \cite{KouRos2020,CheMad2010,MadSch2017,KovRudCia2022,EbeMad2014} and references therein.

The notion of acceptability indices have been also generalised in multiple directions, e.g. by considering a dynamic (conditional) setup, see \cite{BieCiaZha2012,BieCiaIyiRod2013}, or modifying (relaxing) the set of underlying assumptions, see \cite{RosGiaSga2013,Righi2022}. The dual link between risk measures and acceptability indices has been also used to show that the commonly used backtesting metrics, such as the standard exception rate (breach) count for Value-at-Risk, could be also expressed in the language of acceptability indices, not necessarily in the coherent setup, see \cite{MolPit2019,PitSch2022} for details.

In the present paper we expand the theory of acceptability indices and show how one can link the concept introduced in \cite{CheMad2009} with expected utility portfolio optimisation, stochastic control risk-averse objective functions, and generic theory of certainty equivalents, see \cite{BiePli2003,BenTeb1986,FolSch2004} for more background. To this end, we define a novel class of utility-based acceptability indexes which are different in nature from acceptability indices introduced in the literature so far. First, these new objects are defined starting from utility functions, which are a cornerstone concept of mathematical economics. Second, the maps introduced in this paper satisfy a certain \emph{positive inverse homogeneity} property which is structurally different from the \emph{scale invariance} property considered in \cite{CheMad2009}. This apparent small difference has far-reaching consequences. In particular, we show that inverse homogeneous maps are more fit for purpose to model the growth dynamics and quantify the underlying risk aversion. In fact, the starting point of our definition was the family of entropic risk measures indexed by risk-aversion, which is a prominent example of non-coherent (yet convex) risk measures; this family is commonly used in stochastic control, see~\cite{KupSch2009,BiePli1999,PitSte2023} and references therein. 

We want to emphasize that this short research paper is just a first step in analysing the potential of the concept of \emph{positive inverse homogeneity} 
for certainty equivalents. Here, we focus on axiomatic properties, proper formulation of the portfolio optimisation probem, and generic performance assessment. We hope that adequate performance measure axiom modifications combined with process-growth consideration could potentially consolidate the (scale-invariant) acceptability index theory with modern portfolio theory, e.g. where mean-variance optimisation is considered. This also applies to multiple risk-sensitive stochastic control problem formulations, e.g. when risk-sensitive criterion is considered as the extension of the Kelly growth index. See \cite{KalErtAkb2019, MacThoZie2011} and references therein.

The paper is organised as follows: Section \ref{S:preliminaries} defines and analyses, starting from a utility function $U$, a parametrized family of certainty equivalents which are the key building blocks for utility-based acceptability indexes. The core object of this paper, utility-based acceptability index, is formally introduced and studied in Section \ref{section3}. In Section~\ref{S:final.section} we focus on two problems that emerge when studying utility-based acceptability indices. Namely, in Subsection \ref{maxi} we show that the finite-horizon portfolio optimization problem is well-posed under mild assumptions. On the other hand, in Subsection \ref{S:long-run2} we investigate the case of infinite horizon and analyse the dependency between the choice of certainty equivalent, the underlying portfolio dynamics, and performance measure value. 

\section{Preliminaries}\label{S:preliminaries}

Let $(\Omega,\mathcal{F},\mathbb{P})$ be a probability space, let $L^{\infty}:=L^{\infty}(\Omega,\mathcal{F},\mathbb{P})$ denote the set of all essentially bounded and real-valued random variables understood as discounted P\&Ls, future wealth of financial positions or log-returns of wealth. For brevity, we also set $L^1:=L^{1}(\Omega,\mathcal{F},\mathbb{P})$ and  $L^0:=L^0(\Omega,\cF,\bP)$ to denote the set of integrable and all random variables, respectively, and use $\bE[\cdot]$ for the expectation operator.

Our goal is to quantify the performance of a generic position $X\in L^\infty$ following the certainty equivalent approach with additional risk aversion specification. To this end, we consider a concave, strictly increasing, and bounded from above {\it utility function} $U\in C^2(\bR)$; {for brevity, we use notation $U(\infty):=\lim_{x\to\infty}U(x)$}. We define the scaled utility family $(U_{\gamma})_{\gamma>0}$, where $U_{\gamma}(x):=U(\gamma \cdot x)$ for $\gamma>0$ and $x\in\bR$. Note that $U_{\gamma}\in C^2(\bR)$, and $U_{\gamma}$ is concave, strictly increasing, and bounded from above for any $\gamma>0$. Next, we define the family of {\it risk measures} $(\mu_{\gamma})_{\gamma > 0}$, where $\mu_{\gamma}\colon L^0 \to \bR\cup \{+\infty\}$ is given by
\begin{equation}\label{eq:cert.equiv}
\mu_{\gamma}(X) := -U_{\gamma}^{-1}(\bE[U_{\gamma}(X)])=-\tfrac{1}{\gamma}U^{-1}(\bE[U(\gamma X)]);
\end{equation}
In \eqref{eq:cert.equiv}, we use convention $U^{-1}(-\infty):=-\infty$ so that $\mu_{\gamma}(X)$ is well-defined on $L^0$ and takes values in  $(-\infty,\infty]$. For simplicity, from now on, we decided to effectively act on $L^\infty$ but most definitions and properties introduced herein could be restated in $L^0$.

We use parameter $\gamma>0$ to quantify the {\it risk aversion} using the convention that the bigger the value of $\gamma$, the more risk-averse we are. To this end, we say that the utility function $U$ is {\it scale aversion regular}, if for any $X\in L^{\infty}$ and  $0<\gamma_{1}\leq \gamma_{2}<\infty$ we get
\begin{equation}\label{eq:regular}
\mu_{\gamma_{1}}(X)\leq \mu_{\gamma_{2}} (X),
\end{equation}
that is $(u_{\gamma}(\cdot))_{\gamma>0}$ is a non-decreasing function of $\gamma$. Before we introduce the main object of study of this paper, i.e. the acceptability index linked to the family $(\mu_{\gamma})_{\gamma > 0}$, we need to recall the concept of {\it Arrow-Pratt measure of absolute risk aversion}, study how it is related to scale aversion regularity of $U$, and study same basic properties of the family $(\mu_{\gamma})_{\gamma > 0}$. For any $\gamma>0$, the {\it Arrow-Pratt risk aversion function} of $U_{\gamma}$ is given by
\begin{equation}\label{eq:arrow-pratt}
A_{\gamma}(x):=-\frac{U''_{\gamma}(x)}{U'_{\gamma}(x)}=-\gamma\frac{U''(\gamma x)}{U'(\gamma x)}, \quad x\in\bR.
\end{equation}
The link between \eqref{eq:arrow-pratt} and monotonicity of $(\mu_{\gamma})_{\gamma > 0}$ is stated in Proposition~\ref{pr:arrow-pratt}.
\begin{proposition}\label{pr:arrow-pratt} The utility function $U$ is scale aversion regular if and only if the mapping $\gamma\to A_{\gamma}(x)$ is non-decreasing for any $x\in\mathbb{R}$.
\end{proposition}

\begin{proof} 
For completeness, let us expand the proof to the $L^0$ setting.  First, assume that $\gamma\to A_{\gamma}(x)$ is non-decreasing for any $x\in\mathbb{R}$ and let us show~\eqref{eq:regular}. Noting that the mapping \eqref{eq:cert.equiv} is a negative of certainty equivalent and using Proposition 2.47 of \cite{FolSch2004} we know that $\mu_{\gamma_{1}}(X)\leq \mu_{\gamma_{2}}(X)$ for $X\in L^1$. Thus, it remains to show the the inequality holds for non-integrable $X\in L^0$, too.  If $\mu_{\gamma_{2}}(X)=\infty$, then~\eqref{eq:regular} is trivial. If not, then $\bE[X^{-}]<\infty$, where $X^-$ is the negative part of $X$. Thus, for any $n\in\bN$, the random variable $X_{n}:=\min\{X,n\}$ is integrable and consequently we get $\mu_{\gamma_{1}}(X)\leq \mu_{\gamma_{1}}(X_{n})\leq \mu_{\gamma_{2}}(X_{n})$. Now, letting $n\to\infty$ and using the Beppo--Levi Theorem, we deduce that $\mu_{\gamma_{1}}(X)\leq \mu_{\gamma_{2}}(X)$, which concludes this part of the proof. 

Let us now assume~\eqref{eq:regular} for any $X\in L^0$. In particular,~\eqref{eq:regular} must hold for $X\in L^1$, so that we can use again Proposition 2.47 of \cite{FolSch2004} to conclude that $\gamma\to A_{\gamma}(x)$ is non-decreasing for any $x\in\mathbb{R}$.
\end{proof}
Next, for completeness, let us summarise basic properties of the family $(\mu_{\gamma})_{\gamma > 0}$. 

\begin{proposition}\label{pr:risk.measure.prop}
Let $U$ be a scale aversion regular utility. Then the corresponding family of risk measures $(\mu_{\gamma})_{\gamma>0}$ satisfies the following properties for generic $\gamma>0$, $X,Y\in L^\infty$ and $\lambda\in [0,1]$:
\begin{enumerate}
\item Monotonicity, i.e. if $X\geq Y$, then $\mu_{\gamma}(X)\leq \mu_{\gamma}(Y)$;
\item Law-invariance, i.e. if $X\stackrel{law}{=}Y$, then $\mu_{\gamma}(X)= \mu_{\gamma}(Y)$;
\item Normalisation, i.e. $\mu_{\gamma}(c)=-c$ for any $c\in\bR$;
\item Quasi-convexity, i.e. $\mu_{\gamma}(\lambda X+(1-\lambda)Y) \geq \mu_{\gamma}(X) \vee \mu_{\gamma}(Y)$;
\item Parameter monotonicity, i.e. if $\gamma_1 \leq\gamma_2$, then $\mu_{\gamma_{1}}(X)\leq \mu_{\gamma_{2}} (X)$;
\item Risk aversion continuity, i.e. $\gamma\to \mu_{\gamma}(X)$ is continuous;
\item Risk aversion left limit, i.e. $\mu_{\gamma}(X)\to - \bE[X]$ as $\gamma\to 0$;
\item Risk aversion right limit, i.e. if 
{$\frac{U(\gamma x_1)-U(\infty)}{U(\gamma x_2)-U(\infty)}\to 0$} as $\gamma\to\infty$, for all $x_1>x_2$, then $\mu_{\gamma}(X)\to -
\mathrm{ess.}\inf X$ as $\gamma\to\infty$.
\end{enumerate}
\end{proposition}

\begin{proof}
The proof of properties 1.--4. is straightforwad and omitted for brevity; the proof of 5. is a direct consequence of Proposition~\ref{pr:arrow-pratt}. To show risk aversion continuity (6.) is is sufficient to use the dominated convergence theorem. Indeed, noting that $U$ is bounded from above and that for any $X\in L^\infty$ there exists $C<0$ such that $X>C$, we get
\[
\textstyle \sup_{x\in\bR}U(x) \geq U(\gamma_n X)\geq U((\sup_{n\in\bN} \gamma_n) C),
\]
where $(\gamma_n)_{n\in\bN}$ is an arbitrary sequence, such that $\gamma_n>0$, $n\in\bN$, and $\gamma_n\to \gamma>0$. Consequently, by dominated convergence theorem we get  $\lim_{n\to\infty}\mu_{\gamma_n}(X)=\mu_{\gamma}(X)$. 

To prove the risk aversion left limit behaviour (7.), we note that the L'Hospital rule is applicable to the function $\gamma\to U^{-1}(\bE[U(\gamma X)])/\gamma$,
as $\gamma\to 0$, and one can differentiate under the
expectation sign since $X\in L^\infty$. Consequently, we get
\[
\lim_{\gamma\to 0}\mu_{\gamma}(X)=-\lim_{\gamma\to 0}\frac{\bE[U'(\gamma X)X]}{U'(\bE[U(\gamma X)])}=-\frac{U'(0)\bE[X]}{U'(0)}=-\bE[X].
\]

Finally, we prove the risk aversion right limit behaviour (8.). For simplicity and without loss of generality, we can assume that $U(\infty)=0$, which implies $U\leq 0$. Indeed, it is sufficient to note that utilities $U$ and $\tilde U:=U-U(\infty)$ induce the same scaled family of certainty equivalents. Now, note that, as $X\geq \essinf X$, we necessarily have $\mu_\gamma(X)\leq -\essinf X$, for any $\gamma>0$, by the monotonicity property. Arguing by contradiction, let there be $X\in L^\infty$ and $\epsilon>0$ such that $\mu_{\gamma}(X)\nearrow -(\essinf X+\varepsilon$), as $\gamma\to\infty$; recall that $\gamma\to \mu_{\gamma}(X)$ is continuous and increasing due to parameter monotonicity (5.) and risk aversion continuity (6.). In particular, for $x_1:=\essinf X+\varepsilon$, this implies
\begin{equation}\label{eq:U.contr1}
\mathbb{E}[U(\gamma X)]\geq U(\gamma x_1),\quad \gamma\to\infty.
\end{equation}
Now, let $x_2:=\essinf X+\varepsilon/2$. By the definition of essential infimum,
we have $\delta:=\mathbb{P}(X\leq x_2)>0$. Since $U<0$, this implies
\begin{equation}\label{eq:U.contr2}
\mathbb{E}[U(\gamma X)]\leq (1-\delta)U(\gamma\esssup X)+\delta U(\gamma x_2) \leq \delta U(\gamma x_2).
\end{equation}
Now, from \eqref{eq:U.contr1} and \eqref{eq:U.contr2} we get $U(\gamma x_1)\leq \delta U(\gamma x_2)$ which in turn leads to a contradiction as for $x_1>x_2$ we should have $U(\gamma x_1)/U(\gamma x_2)\to 0$, when $\gamma\to\infty$.

\bigskip


\end{proof}
While properties 1.--5. in fact hold for any $X\in L^0$, it might not be the case for 6.--8. Indeed, it is not true, in general, that $\gamma\to \mu_{\gamma}(X)$ is continuous in $\gamma$. To see this it is enough to consider exponential utility $U(x)=-e^{-x}$ and $X=-Z^2$, where $Z$ is standard Gaussian; we get $\mu_{\gamma}(X)=\infty$ for all $\gamma\geq 1/2$ and $\mu_{\gamma}(X)<\infty$ for $0<\gamma<1/2$. That said, since $U$ is bounded from above, Fatou's lemma implies the lower semi-continuity property $\liminf_{n\to\infty}\mu_{\gamma_n}(X)\geq \mu_{\gamma}(X)$ for any $X\in L^0$ and sequence $(\gamma_n)$ converging to $\gamma>0$. Note that 6. holds, if $X$ is assumed to be only bounded from below, see proof above for details. The risk aversion limits also do not hold in general. For left limit (7.), set $U(x)=-e^{-x}$ and $X=-e^Z$, where $Z$ is again standard Gaussian; we get $\mu_{\gamma}(X)=\infty$ for $\gamma>0$ but $\E[X]$ is finite. For the right limit (8.), consider $U$ such that $U(x)=x$ for
$x<0$ and let $X\in L^{\infty}$
be such that $X<0$. We get $\mu_{\gamma}(X)=\bE[X]$ for any $\gamma>0$ so (8.) does not hold
(unless $X$ is constant). Note that (8.) holds true e.g. for the exponential utility.

From Proposition~\ref{pr:arrow-pratt} we see that when $\gamma\to 0$, we are approaching the risk-neutral setting, see~\cite{Whi1990}. The introduced scaling is often used in the risk-sensitive stochastic control problems, where $\gamma$ is used to set up the relation between mean and variance in the objective function, see e.g.~\cite{BiePli2003}. For completeness, let us now show some examples of popular utility functions and check if they are scale aversion regular.

\begin{example} {\rm (Exponential utility) Let us define the standard exponential utility by setting
\begin{equation}\label{eq:U.exponential}
U(x):=-e^{-x},\quad x\in\bR.
\end{equation}
It is easy to check that $U\in C^2(\bR)$ and $U$ is strictly increasing, concave, and bounded from above. It is also well known that the exponential utility has a constant relative risk aversion, i.e. we have
\[
A_{\gamma}(x)=\gamma,\quad x\in\bR,
\]
so that $U$ is scale aversion regular, see~\cite{FolSch2004}. Moreover, for any $\gamma>0$, we get
\begin{equation}\label{eq:entropic.risk}
\mu_{\gamma}(X)=\tfrac{1}{\gamma}\ln(\bE[e^{-\gamma X}]),\quad X\in L^\infty,
\end{equation}
which shows that exponential utility defines, through \eqref{eq:cert.equiv}, the family of \emph{entropic risk measures}. It is worth noting that the entropic risk measure given in \eqref{eq:entropic.risk} is cash additive, i.e. for any constant $a\in\bR$ and $X\in L^0$ we get $\mu_{\gamma}(X+a)=\mu_{\gamma}(X)-a$. We refer to \cite{KupSch2009} for more details on this family and its unique properties.}
\end{example}

\begin{example}\label{pl} {\rm (Power-like utility functions) Let us fix $\alpha>0$, $\beta\geq 2$, and set
\[
U(x):=
\begin{cases}
\frac{-(1+x)^{-\alpha}}{\alpha(\alpha+1)} &
\mbox{ for }x\geq 0,\\
\frac{-(1-x)^{\beta}}{\beta(\beta-1)}+
\frac{\beta-\alpha-2}{(\alpha+1)(\beta-1)}x+\frac{1}{\beta(\beta-1)}-
\frac{1}{\alpha(\alpha+1)} &\mbox{ for }x<0.
\end{cases}
\]
It is easy to check that $U\in C^2(\bR)$ and $U$ is strictly increasing, concave, and bounded from above. By direct calculations we also get
\[
A_{\gamma}(x)=
\begin{cases}
\frac{\gamma(\alpha+1)}{1+\gamma x}  & \mbox{ for }x\geq 0,\\
 \frac{\gamma (1-\gamma x)^{\beta-2}}{\frac{(1-\gamma x)^{\beta-1}}{\beta-1}+
\frac{\beta-\alpha-2}{(\alpha+1)(\beta-1)}} & \mbox{ for }x<0.
\end{cases}
\]
It is straightfoward to check that the function $\gamma\to A_{\gamma}(x)$ is non-decreasing for any $x\in\bR$, so that $U$ is scale aversion regular.}
\end{example}

\begin{example}{\rm (Iterated exponential utility) Let us define the iterated exponential utility as
\[
U(x):=-\exp(e^{-x}),\quad x\in\mathbb{R}.
\]
Is is easy to check that $U\in C^2(\bR)$ and $U$ is strictly increasing, concave, and bounded from above. Also, we get  
\[
A_{\gamma}(x)=\gamma(e^{-\gamma x}+1),\quad x\in\bR.
\]
Thus, the function $A_{(\cdot)}(x)$ fails to be non-decreasing, and consequently $U$ is not scale aversion regular.}
\end{example}
At the end of this section we present a few remarks linked to certainty equivalents defined in \eqref{eq:cert.equiv}.

\begin{remark}[Locality of the risk aversion] We note that risk aversion regularity
is a \emph{local} property in the sense that if
it is checked on, say, a collection
of intervals covering $\mathbb{R}$, then it is
satisfied globally. This means that every $U\in C^2(\bR)$ 
that is a \emph{piecewise} power function or an
exponential function (or the logarithm function) satisfies risk aversion regularity.
Example \ref{pl} is a specific case, where $U$ is a power function on both the positive and negative half-lines.
\end{remark}

\begin{remark}[Scale invariant certainty equivalent]
In the limiting case of a linear utility, scaling does not change anything, i.e. for $U(x)=x$ we get $\mu_{\gamma}(X)=\mu_{1}(X)$ for any $\gamma>0$. On $L^0$, the family of utilities $U(x)=ax+b$, for $a>0$ and $b\in\bR$ is the only family that produces a scale-invariant certainty equivalent, see e.g.~\cite{PraDro2013}. 
\end{remark}

\begin{remark}[Stochastic control of MDPs with average-cost criteria based on certainty equivalents]
Results of \cite{CavHer2016} show that if $U$ behaves asymptotically like the exponential utility $-e^{-\gamma x}$ then the corresponding control problems admits the same solutions as for the exponential utility, for a suitable one-step reward functions. In other words, for any utility $U$ the quantity $\lambda:= \lim_{x\to\infty}A_{\gamma}(x)$ is a key to the long-run optimal portfolio choice problems under the settings considered e.g. in \cite{CavHer2016} or \cite{Ste2023}. 
\end{remark}

\begin{remark}[Affine transforms of utility measure]\label{rem:affine}
One may wonder about another type of scaling: $\tilde U_{\gamma}(x):=\gamma U(x)$, for $\gamma>0$. This is essentially different from the scaling introduced in this paper and has different purpose. In particular, $\tilde{U}_{\gamma}$ has the same risk-aversion and it induces the same preference order as $U$. Additive shifts applied to $U_{\gamma}$ also do not change the preference order, see \cite{FolSch2004}.
\end{remark}

\begin{remark}[Certainty equivalents, mean value principles, and time-conistency]
Certainty equivalents and their properties were extensively studied in the insurance literature, see e.g. \cite{Ger1979}, where they are often referred to as {\it mean value principles}. They also play an important role in the dynamic setting as they are the only maps which are strongly time-consistent for any filtration, see e.g.~\cite{Ger1979,KupSch2009}.
\end{remark}

\begin{remark}[{Certainty equivalents as risk measures}]
In Proposition~\ref{pr:risk.measure.prop}, we summarised the basic axiomatic properties of the family $(\mu_{\gamma})$ which is essentially a family of certainty equivalents. One might ask a question for which $U$ the corresponding family $(\mu_{\gamma})_{\gamma>0}$ satisfies cash-additivity or convexity. It turns out that in both cases the exponential utility is the only utility function for which this properties are satisfied, see \cite{Mul2007} for details.
\end{remark}

\begin{remark}[{Cash additivity and Optimized Certainty Equivalents}]
One can consider an alternative to family $(\mu_{\gamma})_{\gamma>0}$ using the concept of optimised certainty equivalents. Namely, following \cite{BenTeb1986}, one can define a family of maps $(\rho_{\gamma})_{\gamma>0}$ given by $\rho_{\gamma}(X)=\sup_{c\in\bR}\{c+\bE[\tilde U_{\gamma}(X-c)]\}$, where $\tilde U_{\gamma}(\cdot):=U(\gamma\cdot)/\gamma$. It turns out, that while this family induces different order when compared to classical certainty equivalent,  and it is cash-additive, it also satisfies certain risk aversion ordering, see Section 2.4 in \cite{BenTeb2007} for details. We also note that if one inherently requires cash-additivity, it can be in fact induced directly on certainty equivalents by considering the modified family $(\hat\mu_{\gamma})_{\gamma>0}$ given by
\[
\hat\mu_{\gamma}(X):=\sup_{c\in\bR}\{c+\mu_{\gamma}(X-c)]\}.
\]
In particular, note that the family $(\hat\mu_{\gamma})_{\gamma>0}$ maintains most of the properties listed in Proposition~\ref{pr:risk.measure.prop}, including the risk-aversion parameter monotonicity.
\end{remark}


\section{Utility-based acceptability index}\label{section3}
In this section we introduce the main object of study in this paper, i.e. the utility-based acceptability index based on the family introduced in \eqref{eq:cert.equiv}. We follow the robust duality framework introduced in \cite{CheMad2009} in which the acceptability index was defined as a risk acceptance induced supremum over a parameter-increasing family of risk measures.

\begin{definition}[Utility-based acceptability index]
Let $U$ be a scale aversion regular utility. Then, the {\it acceptability index} $\alpha\colon L^{0} \to [0,+\infty]$ based on utility $U$ is given by
\begin{equation}\label{eq:a.index}
\alpha(X):=\sup\{\gamma>0:\mu_{\gamma}(X)\leq 0\},\quad X\in L^0,
\end{equation}
where $(\mu_{\gamma})_{\gamma>0}$ is a family defined in~\eqref{eq:cert.equiv}; in \eqref{eq:a.index} we use convention $\sup\emptyset:= 0$.
\end{definition}

For brevity, we often refer to utility-based acceptability indexes defined in \eqref{eq:a.index} simply as UAI. Also, as in the previous section, we restrict ourselves to the $L^\infty$ domain. In this setting, Cherny and Madan state and prove eight properties for the acceptability indices based on coherent risk measures, see \cite{CheMad2009}. Namely, this refers to: {\it monotonicity}, {\it arbitrage consistency}, {\it quasi-concavity}, {\it law invariance}, {\it expectation consistency}, {\it Fatou property}, {\it consistency with second-order stochastic dominance}, and {\it scale invariance}. The first seven of these properties also hold for UAIs as Proposition \ref{sevenprop} below shows. On the other hand, the scale invariance property, i.e. $\alpha(\lambda X)=\alpha(X)$ for $\lambda>0$ and $X\in L^{\infty}$, does not hold in our case due to the lack of positive homogeneity of the underlying risk measure. That said, the UAI introduced in this paper satisfies the \emph{inverse homogeneity property}. As pointed out in the concluding section of \cite{CheMad2009}, scale invariance excludes the use of their acceptability index for measuring preferences. This is no longer the case in the present setting -- that is why $\alpha$ can be used for portfolio choice problems, see Section \ref{maxi} for details.


\begin{proposition}[Axiomatic properties of UAIs]\label{sevenprop}
Let $U$ be a scale aversion regular utility. Then, the acceptability index $\alpha$ given in \eqref{eq:a.index} satisfies the following properties for $X,Y\in L^\infty$:
\begin{enumerate}
\item Monotonicity, i.e. if $X\leq Y$, then $\alpha(X)\leq \alpha (Y)$;
\item Arbitrage consistency, i.e. $X\geq 0$ if and only if $\alpha(X)=\infty$;
\item Quasi-concavity, i.e. $\alpha(\lambda X+(1-\lambda)Y)\geq \alpha(X)\wedge \alpha(Y)$;
\item Law invariance, i.e. if $X\stackrel{law}{=}Y$, then $\alpha(X)=\alpha(Y)$;
\item Expectation consistency, i.e. if $E[X]< 0$, then $\alpha(X)=0$ and if $E[X]>0$ then $\alpha(X)>0$.
\item Fatou property, i.e. if $\alpha(X_n)\geq x$, for $n\in\bN$ and $X_n\in L^\infty$, such that $X_n\to X$ (a.s.), then $\alpha(X)\geq x$.
\item Consistency with second-order stochastic dominance, i.e. if $X\preceq_{2} Y$, then $\alpha(X)\leq \alpha(Y)$.

\item Inverse positive homogeneity, i.e. for $\lambda>0$ we get $\alpha(\lambda X)=\lambda^{-1}\alpha(X)$.
\end{enumerate}
\end{proposition}
\begin{proof}
For completeness, let us present the whole proof point by point.

\medskip

\noindent {\it Monotonicity} (1.): the proof of monotonicity follow directly from the (inverse) monotonicity of $u_{\gamma}(\cdot)$, for any $\gamma>0$, as well as parameter monotonicity of the family $(u_{\gamma}(\cdot))_{\gamma>0}$, see Proposition~\ref{pr:risk.measure.prop}.

\medskip

\noindent {\it Arbitrage consistency} (2.): assume that $X\geq 0$. Then, for any $\gamma>0$, we get $\mu_{\gamma}(X)\leq \mu_{\gamma}(0)=0$. Consequently, directly form \eqref{eq:a.index}, we get $\alpha(X)=\infty$. On the other hand, let us assume that there exists $X\in L^0$ such that $\alpha(X)=\infty$ and $\bP[X<0]>0$. We know that $\alpha(X)=\infty$ implies that for any $\gamma>0$ we get
\begin{equation}\label{eq:cert1}
U^{-1}(\bE[U(\gamma X)])\geq 0.
\end{equation}
On the other hand, from $\bP[X<0]>0$, noting that $U$ is bounded from above, strictly monotone, and concave, we get $\bE[U(\gamma X)]\to -\infty$, as $\gamma\to\infty$, which contradicts \eqref{eq:cert1}.

\medskip

\noindent {\it Quasi-concavity} (3.):  for any $\gamma>0$, the risk measure $u_{\gamma}(\cdot)$ is quasi-convex, see Proposition~\ref{pr:risk.measure.prop}. Consequently, we get
\begin{align*}
\alpha(\lambda X+(1-\lambda Y)&\geq \sup\{\gamma>0\colon \mu_{\gamma}(\lambda X+(1-\lambda)Y)\leq 0 \}\\
&\geq \sup\{\gamma>0\colon \mu_{\gamma}(X) \vee \mu_{\gamma}(Y)\leq 0 \}\\
& \geq \alpha(X) \wedge \alpha(Y).
\end{align*}

\medskip

\noindent {\it Law-invariance} (4.):  This follows directly from the law invariance of $u_{\gamma}(\cdot)$, see also~\cite{KupSch2009}.

\medskip

\noindent {\it Expectation consistency} (5.):  assume $\bE(X)< 0$. Then by Jensen's inequality and strict monotonicity of $U$, for any $\gamma>0$ we have
\[
U(0)> U(\bE[\gamma X]) \geq \bE[U(\gamma X)].
\]
Taking $U^{-1}$ on both sides and multiplying by $-1/\gamma$ we get that $\mu_{\gamma}(X)>0$ for $\gamma>0$ which implies that $\alpha(X) = \sup\{\gamma>0\colon \mu_{\gamma}(X)\leq 0 \}=\sup\{\emptyset\}=0$. Conversely, assume that $\bE[X]>0$. From Proposition~\ref{pr:risk.measure.prop} we know that the mapping $\gamma\to \mu_{\gamma}(X)$ is continuously decreasing to $-\bE(X)$. Consequently, there exists $\gamma>0$ such that $\mu_{\gamma}(X)\leq 0$ which immediately implies $\alpha(X)>0$, by the lower semi-continuity of $\gamma\to \mu_{\gamma}(\cdot)$.

\medskip

\noindent {\it Fatou property} (6.):  Recalling that $U$ is bounded from above, the statement follows directly from Fatou's lemma. 

\medskip

\noindent {\it Consistency with second-order stochastic dominance} (7.):  We know that $X$ second-order stochastically dominates $Y$ if and only if for any increasing and concave function $V\colon\bR\to\bR$ we get $\bE[V(X)]\leq \bE[V(Y)]$. Consequently, the claim follow by simply setting $U=V$.

\medskip

\noindent {\it Inverse positive homogeneity} (8.):  For any $\lambda>0$ we get
\begin{align*}
\alpha(\lambda X) &=\sup\{\gamma>0\colon \mu_{\gamma}(\lambda X)\leq 0 \}\\ 
& =\sup\{\gamma>0\colon \lambda\mu_{(\lambda\gamma)}(X)\leq 0 \}\\
& =\lambda^{-1} \sup\{\lambda\gamma\in \bR_+\colon \mu_{(\lambda\gamma)}(X)\leq 0 \}\\
& =\lambda^{-1} \alpha(X).
\end{align*}
\end{proof}
As before, for simplicity, we formulated the results in the restricted $L^\infty$ setting; note that most properties (i.e. 1.--4. and 8.) effectively hold on $L^0$. As already stated before, UAI satisfies the  {\it inverse positive homogeneity property} ($\alpha(\lambda X)=\lambda^{-1}\alpha(X)$) rather than {\it scale invariance property} ($\alpha(\lambda X)=\alpha(X)$), which was the defining property stated in the coherent acceptability index framework in \cite{CheMad2009}. This makes the indices introduced in this paper essentially different from the classical ones. In particular, we believe that our definition is better suited for certain stochastic control problems. In particular, in the stochastic control literature, the mapping $\alpha$ is typically applied to the log-return (or log-growth) of the position rather than to the P\&L vector, see e.g.~\cite{PitSte2016,BieCiaPit2013}. From mathematical viewpoint, this is done to make the underlying dynamic close to the MDP framework in which we can control the logarithm dynamics of the process. In this framework, the position scale invariance is effectively pre-assumed since log-return itself is scale invariant and the role of inverse positive homogeneity relates to control over process log-return rather than net size scaling, see Remark~\ref{rem:log.return} for details.

\begin{remark}[{Inverse scaling in the log-return setting}]\label{rem:log.return}
The inverse positive homogeneity might appear naturally in the log-return setting when standard performance (optimisation) measures are used. To better understand this, let us consider a simplified setting within the classic mean-variance optimisation framework with log-return control and objective criterion given by
\[
\mu_{\gamma}(\ln X)\approx \bE[\ln X] +\tfrac{\gamma}{2}\textrm{Var}[\ln X],
\]
for any risk aversion $\gamma>0$. Note that $(\mu_{\gamma})$ could be seen as a second order Taylor approximation of the risk-sensitive control objective criterion which in turn is a certainty equivalent (for exponential utility), see \cite{BiePli2003} for details. The acceptability index dual to this family identifies the biggest risk aversion $\gamma>0$ for which position $\ln X$ is acceptable in terms of mean-to-variance ratio. This risk-to-reward criterion is inverse positive homogeneous by design due to mean and variance properties.
\end{remark}

\begin{remark}[{Star-shaped acceptability indices}] 
The family of  \emph{star-shaped} acceptability indices has been recently introduced in \cite{Righi2022}. In a nutshell, this family satisfies the property $\alpha(\lambda X)\leq \alpha(X)$ for all $X$ and for $\lambda\geq 1$, which is weaker than scale-invariance. It is clear 
directly from the inverse positive homogeneity property
that the utility-based acceptability indices introduced in this paper are star-shaped.
\end{remark}

\section{Measuring portfolio performance using utility-based acceptability indices}\label{S:final.section}

Let $(\Omega,\mathcal{F},(\mathcal{F}_{t})_{t\in \bT},\bP)$ be a continuous-time filtered probability space where the time horizon is either finite, i.e. $\bT=[0,T]$ for some $T\in\bR$, or infinite, i.e. $\bT=[0,\infty)$. Let $(S_{t})_{t\in\bT}$ be an adapted $\mathbb{R}^{d}$-valued c\`adl\`ag semimartingale, $d\in\bN$.
We assume, without loss of generality, that the investor in consideration has initial wealth $V_0:=1$, so that we can associate portfolio log-return with its log-growth. Denote by $\Phi$ the set of $\mathbb{R}^{d}$-valued predictable processes such that for any $(\phi_{t})_{t\in\bT}\in \Phi$ the corresponding portfolio value process $(V_t(\phi))_{t\in\bT}$, given by
\[
V_{t}(\phi):=1+\int_{0}^{t} \phi_{u}\, dS_{u},\quad t\in\bT,
\]
is non-negative, i.e. $V_{t}(\phi)> 0$ almost surely, for $t\in\bT$. In this section, given scale aversion regular utility $U$, we are interested in calculating or maximising the value 
\begin{equation}\label{eq:performance}
\alpha(\ln V_t(\phi) - \ln G_t),\quad t\in\bT,
\end{equation}
where $G:=(G_t)_{t\in\bT}$ corresponds to a pre-specified benchmark value position. The value \eqref{eq:performance} is essentially quantifying the risk aversion which makes the benchmarked return $\ln V_t(\phi) - \ln G_t$ acceptable at $t\in\bT$. In particular, in the finite time setting, we are often interested in maximising the terminal value $\alpha(\ln V_T - \ln G_T)$ while in the infinite time setting we want to measuring the limit performance of the mapping $t\to \alpha(\ln V_t(\phi) - \ln G_t)$, e.g. by considering the limit inferior with respect to time. The position $G$ might correspond to a position we want to outperform or hedge. For instance, one might consider a risk-free position $G_t=e^{\lambda t}$, where $\lambda\in \bR$ is some constant growth rate we want to outperform.

\subsection{Finite-time horizon portfolio performance maximisation}\label{maxi}
In this section we fix $\bT=[0,T]$ with a terminal time $T>0$. The main goal of this section is to show that under suitable assumptions the problem
\begin{equation}\label{eq:finite.obj}
J(\phi,G):=\alpha(\ln(V_T(\phi))-\ln(G_T))\to \max
\end{equation}
can be solved; the maximum in \eqref{eq:finite.obj} is taken with respect to $\phi\in \Phi$, while $G$ and $\alpha$ are pre-fixed. To this end, we make the following assumptions:

\medskip

\begin{enumerate}
\item[(\namedlabel{A.1}{A.1})] (Utility-based performance measure as objective criterion) The objective function $\alpha\colon L^{0} \to [0,+\infty]$ is an UAI for which the underlying utility satisfies $U(0)= 0$.

\item[(\namedlabel{A.2}{A.2})] (Local martingale property) There exists $\bQ\sim \bP$ such that $S$ is a $\bQ$-local martingale and we have $\bE\left[U^*\left(d\bQ/d\bP\right)\right]<\infty$, where
\[
U^*(y):=\sup_{x\in\mathbb{R}}[U(x)-xy],\quad y>0,
\]
is the convex conjugate of $U$.
\item[(\namedlabel{A.3}{A.3})] (Finite benchmark position utility) We assume that $G_T\in L^0$ is such that $G_T>0$ and $\bE[U(-\ln^-(G_T))]>-\infty$.
\end{enumerate}

\medskip

Assumption \eqref{A.1} is made to recall the optimisation context; $U(0)=0$ can always be achieved by adding a
constant to the utility functions; 

Assumptions \eqref{A.2} and \eqref{A.3} are standard conditions in the utility maximization literature, see \cite{Ras2018}. Note that, contrary to most studies on utility maximization, we make no additional assumption on the asymptotic elasticity of $U$ (at $-\infty$); see \cite{Sch2001}, where a thorough discussion of this concept in utility maximization is given.


We are now ready to present the main result of this section, i.e. Theorem~\ref{maxim}; the proof of this theorem is based on techniques adapted from \cite{Ras2018}.

\begin{theorem}\label{maxim} Let \eqref{A.1}, \eqref{A.2}, and \eqref{A.3} hold. 
Then there is $\phi^{*}\in \Phi$
such that \begin{equation}\label{problem}
J(\phi^*,G)=\sup_{\phi\in\Phi}J(\phi,G)
\end{equation}
Furthermore, $\sup_{\phi\in\Phi}J(\phi,G)=\infty$ if and only if there is $\phi\in\Phi$ such that $V_T(\phi)\geq G_T$, that is, $G_T$ can be superhedged.
\end{theorem}

\begin{proof}
For brevity, we set $\mathcal{S}:=\sup_{\phi\in\Phi}J(\phi,G)$; note that $\mathcal{S}$ could be infinite. If $\mathcal{S}=0$, then an arbitrary $\phi^{*}\in \Phi${}
satisfies \eqref{problem}. Thus we may assume from now on that $\mathcal{S}\in (0,\infty]$. Let us choose
a sequence $(\phi_{n})$, where $\phi_{n}\in\Phi$, such that 
\[
\gamma_{n}:=\alpha(\ln(V_T(\phi_{n}))-\ln(G_T))
\]
is a finite, positive, and increasing sequence, and we have $\gamma_n \to \mathcal{S}$ as $n\to\infty$. For brevity, we also introduce notation $Y_{n}:=\ln(V_T(\phi_{n}))-\ln(G_T)$, $n\in\bN$. By the lower semi-continuity of $\gamma\to \mu_{\gamma}(\cdot)$ we get
\begin{equation}\label{chanel}
\mu_{\gamma_n}(Y_n)\leq 0,
\end{equation}
which also implies 
\begin{equation}\label{chanel1}
\mu_{\gamma_n}(Y_m)\leq \mu_{\gamma_n}(Y_m)\leq 
0,\ m\geq n.
\end{equation}
By Fenchel inequality, $U(0)=0$, and 
\eqref{chanel}, we get
\begin{eqnarray*}
\mathbb{E}_{\bQ}[\gamma_{n} {Y}_{n}^{-}] &\leq& \mathbb{E}[U^*(d\bQ/d\bP)]-\mathbb{E}[U(-\gamma_{n}{Y}_{n}^{-})]\\
&=& \mathbb{E}[U^*(d\bQ/d\bP)]-\mathbb{E}[U(\gamma_{n}{Y}_{n})]+ \mathbb{E}[U(\gamma_{n}{Y}_{n}^{+})]       \\
&\leq& \mathbb{E}[U^*(d\bQ/d\bP)]-U(0)+U(\infty),
\end{eqnarray*}
which in turn implies
\begin{equation}\label{qbound1}
\sup_{n\in\bN}\mathbb{E}_{\bQ}[\gamma_{0}{Y}_{n}^{-}]\leq \sup_{n\in\bN}\mathbb{E}_{\bQ}[\gamma_{n}{Y}_{n}^{-}]<\infty.
\end{equation} 
Also, note that
\[
\mathbb{E}_{\bQ}[-\ln^{-}(G_T)]\leq \mathbb{E}[U^*(d\bQ/d\bP)]-\mathbb{E}[U(-\ln^{-}(G_T))],
\]
where the right-hand side is finite by our assumptions.
Thus, it follows that
\[
\sup_{n\in\bN}\mathbb{E}_{\bQ}[\ln^{-}(V_T(\phi_{n}))]<\infty.
\]
Now, since $t\to \int_{0}^{t}{\phi}_{u}(n)\, dS_{u}$
is a stochastic integral process that is bounded from below, it is a $\bQ$-supermartingale. Hence, 
the process $t\to \ln\left(\int_{0}^{t}{\phi}_{u}(n)\, dS_{u}\right)$ is also a $\bQ$-supermartingale which in turn implies
\[
\sup_{n\in\bN}\mathbb{E}_{\bQ}\left[\ln^{+}\left(1+\int_{0}^{t}{\phi}_{u}(n)\, dS_{u}\right)\right]<\infty.{}
\] 
Consequently, it follows that the
sequence $\int_{0}^{T}{\phi}_n(u)\, dS_{u}$, $n\in\mathbb{N}$, is bounded in $\bQ$-probability. A straightforward implication of Lemma 9.8.1 from \cite{DelSch2006} shows that 
there are convex combinations $\tilde{\phi}_{n}:=\sum_{j=n}^{m_{n}}c_{j}(n){\phi}_{j}$, where $c_{j}(n)\geq 0$ and $\sum_{j=n}^{m(n)}c_{j}=1$,
such that
\begin{equation}\label{eq:proof:1}
\int_{0}^{T}\tilde{\phi}_{n}(u)\, dS_{u}\to J^{*},\quad \textrm{$\bQ$-almost surely},
\end{equation}
and hence also $\bP$-almost surely, where $J^{*}$ is a (finite) random variable. Note that each stochastic integral in \eqref{eq:proof:1} is bounded from below by $-1$ and hence, by Theorem 9.4.2 in \cite{DelSch2006}, we know that there exists $\phi^{*}$, such that
\[
J^{*}\leq \int_{0}^{T}{\phi}^{*}(u)\, dS_{u}.
\]
It should be emphasized that in the statement of
Theorem 9.4.2 in \cite{DelSch2006} the process $(S_t)$ is assumed to be bounded but this is actually not required for its validity, see the bottom of page 294 in \cite{DelSch2006} for details. 

Not that at this point it is not yet clear whether 
\begin{equation}\label{moa}
\int_{0}^{T}{\phi}^{*}(u)\, dS_{u}>-1 
\end{equation}
almost
surely. Next, from concavity of $U$ and the logarithm function and from \eqref{chanel1}, we get
\[
\mathbb{E}\left[U\left(\gamma_{n}\left[\ln\left(\int_{0}^{T}\tilde{\phi}_{n}(u)\, dS_{u}\right)-\ln(G_T)\right]\right)\right]\geq U(0)
\]
for all $n$.
Now let us first consider two cases: (a) $\mathcal{S}<\infty$; (b) $\mathcal{S}=\infty$.
\medskip

\noindent {\it Case a)} $\mathcal{S}<\infty$. Fatou's lemma yields
\[
\mathbb{E}\left[U\left(\mathcal{S}\left[\ln\left(\int_{0}^{T}{\phi}^{*}(u)\, dS_{u}\right)-\ln(G_T)\right]\right)\right]\geq 
\mathbb{E}\left[U(\mathcal{S}(\ln(J^{*})-\ln(G_T))\right]\geq U(0),
\]
which also guarantees $\phi^*\in \Phi$, that is,
\eqref{moa}.
Thus, by the definition of $\alpha$, we get $\mathcal{S}\leq \alpha(\ln(V_T(\phi^{*}))-\ln(G_T)))$. Since the reverse equality comes directly from the definition of $\mathcal{S}$, we arrive at \eqref{problem}.

\medskip

\noindent {\it Case b)} $\mathcal{S}=\infty$. Noting
 $U\leq 0$, we know that for any $\varepsilon>0$ 
\begin{equation}\label{eq:th1:2}
\mathbb{E}[U(\gamma_{n}(\ln(V_T(\phi^{*}))-\ln(G_T)))]\leq U(-\gamma_{n}\varepsilon)\mathbb{P}[\ln(V_T(\phi^{*}))-\ln(G_T)\leq -\varepsilon]+U(\infty).
\end{equation}
If the probability on the right-hand side of \eqref{eq:th1:2} is positive for any $\varepsilon>0$, the left-hand side must converge to $-\infty$, as $n\to\infty$. This would lead to contradiction since the left-hand side is bounded from below by $U(0)$. Consequently, we must have $\ln(V_T(\phi^{*}))-\ln(G_T)\geq 0$ almost surely. This implies $\alpha(\ln(V_T(\phi^{*}))-\ln(G_T)))=\infty$, and concludes the proof of  \eqref{problem}. 

Finally, if $V_T(\phi^{\dagger})\geq G$ for some $\phi^{\dagger}\in\Phi$, then clearly 
$\mu_{\gamma}((\ln(V_T(\phi^{\dagger}))-\ln(G_T)))\leq 0$ for all $\gamma>0$, and consequently  
$\mathcal{S}=\alpha(\ln(V_T(\phi^{\dagger}))-\ln(G_T))=\infty$. All the statements are now shown.
\end{proof}

\begin{remark}[{Portfolio performance maximisation for non-utility based acceptability indices}] In the proof of Theorem \ref{maxim}, we exploited the connection with utility maximization problems and relied on the techniques developed for them, see \cite{Ras2018}. Consequently, those proof techniques might not transfer directly if one is interested in a more general (monotone) families of performance measures $(\mu_{\gamma})$ that constitute performance measures.
\end{remark}

\subsection{Infinite-time horizon portfolio performance measurement}\label{S:long-run2} 
In this section we assume $\bT=[0,\infty)$  and consider a benchmark process given by $G_t:=e^{\lambda t}$, $t\in\bT$, for a fixed $\lambda\in \bR$. In this setting, given growth rate $\lambda\in \bR$, we are interested in measuring the long-run performance of a portfolio by setting
\begin{equation}\label{eq:long.run.perf}
J_{\infty}(\phi,\lambda):=\liminf_{T\to\infty} \alpha\left(\ln V_T(\phi)-\lambda T\right).
\end{equation}
Recalling the $\alpha$ is given as an index dual to the family of certainty equivalents via risk-aversion parametrisation, the value of \eqref{eq:long.run.perf} might be interpreted as the level of risk aversion under which the investor is indifferent between investment in the risky position $(V_t(\phi))_{t\in\bT}$ and risk-free investment with growth rate $\lambda$. In particular, we can recover investor's risk aversions from the market by comparing their investments to (benchmark) risk-free investments. {In the following examples we show that the long-run performance is inherently linked to the interaction between the process dynamics, benchmark dynamics, and the choice of the underlying utility. This shows that one should be careful when introducing, measuring, or even statistically extracting the risk aversion. To ease the notation, given a pre-fixed strategy $\phi\in\Phi$ and $\lambda\in\bR$, and a dyadic time-grid on which we observe the values, say for $t\in\bN$, we denote the corresponding sequence of one-step portfolio log-returns and cumulative benchmarked growth by $r_t :=\ln V_t(\phi)/V_{t-1}(\phi)$ and $S_t:=\ln V_t(\phi)-\lambda t=(r_1-\lambda)+\ldots+(r_t-\lambda)$, respectively, for $t\in\bN$.}

\begin{example}[Recovering standard risk-aversion under i.i.d. setting]\label{ex:ent.iid}
Let us assume that we are given the exponential utility $U$ defined in \eqref{eq:U.exponential}. For simplicity, let us assume that $\phi$ is such that $(r_t)_{t\in\bT}$ is an i.i.d. sequence. For exponential utility, the corresponding certainty equivalent is additive for independent random variables, so that we get
\begin{align}
\alpha(\ln V_T(\phi)-\lambda T) & =\inf\left\{\gamma\in \bR\colon \sum_{t=0}^{T-1}\mu_{\gamma}\left(\ln \frac{V_{t+1}(\phi)}{V_t(\phi)}-\lambda\right)\geq 0\right\}\nonumber\\
&=\inf\left\{\gamma\in \bR\colon \mu_{\gamma}\left(r_1-\lambda\right)\geq 0\right\}.\label{eq:enri}
\end{align}
For instance, if $r_t\sim N(m,\sigma)$, where
$m>\lambda$ and $\sigma>0$,
then
\[
\alpha(\ln V_T(\phi)-\lambda T)=2(m-\lambda)/\sigma^2,
\]
as $\mu_\gamma (r_1-\lambda)=-(m-\lambda)+\frac{\gamma}{2}\sigma^2$; note that this is the value of the standard risk-aversion coefficient in the mean-variance optimisation.  In general, we can use log-returns as a sample to estimate entropic risk for any $\gamma>0$ and use \eqref{eq:enri} to recover the corresponding risk-aversion. Unfortunately, for non-exponential utility we do not have cash-additivity property, which makes the risk aversion extraction problem more challenging, both theoretically and computationally. 
\end{example}


\begin{example}[Long-run performance for non-exponential utility]\label{rem:unbu} 
The choice of a non-exponential utility in a (near) i.i.d. setting may lead to a situation when we always get $J_{\infty}(\cdot,\lambda)=\infty$, i.e. the acceptability index could be somewhat too lenient on a long time horizon. To show this, let us consider the modified exponential utility $U:\bR\to \bR_{-}$ given by
\[
U(x):=
\begin{cases}
-e^{-x} & \mbox{ for }x\geq 0\\
x-1 &\mbox{ for }x<0
\end{cases}.
\]
Let $\phi$ be such that $(r_t)_{t\in\bT}$ is an i.i.d. sequence, and $r_t\sim N(m,1)$ for $m>\lambda$. In this setting we always get $J_{\infty}(\phi,\lambda)=\infty$. It suffices to show that for each $\gamma>0$ we get $\mu_{\gamma}(S_T)\to \infty$ as $T\to\infty$, which in turn is guaranteed by $\mathbb{E}[U(\gamma S_T)]\to 0$. Since $S_T\sim N((m-\lambda)T,\sqrt{T})$, we know that $S_T\to\infty$ in probability, as $T\to\infty$. Consequently, we get
\begin{equation}\label{eq:ex.last.1}
\bE\left[U(\gamma S_T)\right]= \bE\left[U(\gamma S_T)1_{\{S_T>0\}}\right]- \int_{-\infty}^0 U(\gamma x)e^{-(x-(m-\lambda)T)^2/2T}\, dx.
\end{equation}
In \eqref{eq:ex.last.1}, the first term tends to $0$ by the dominated convergence theorem. The second term is equal to  $\int_{0}^{\infty} (\gamma x +1) e^{-(x+(m-\lambda)T)^2/2T}\, dx$. The function inside the integral is dominated by $(\gamma x+\gamma)e^{-x(m-\lambda)}e^{-(m-\lambda)^2 T /2}$, which implies that the second term in \eqref{eq:ex.last.1} is bounded from above by some constant multiplied by $\gamma e^{-(m-\lambda)^2 T /2}$. This proves our claim, i.e. we get $\mathbb{E}[U(\gamma S_T)]\to 0$ as $T\to\infty$. Note that this argumentation in in fact valid for any utility $U$ for which left limit tends to $-\infty$ at a high power rate.
\end{example}

\begin{example}[Process memory and its impact on the long-run performance] 
Let us assume that we are given the exponential utility. Let us assume that $\phi$ is such that $(r_t)_{t\in\bN}$ is a stationary Gaussian process with mean $m-\lambda>0$.
Let us denote the mean of $S_T$ by $\mu_T:=(m-\lambda)T$ and its variance by $\sigma_T^2:=\mathrm{Var}(S_T)$.
The quantity $\sigma_T^2/T$ typically has a limit,
see for instance Theorem 18.2.2 and Theorem 18.5.2
in \cite{IbrLin1971}. Assume that $\sigma_T^2/T\to \sigma^2\in [0,\infty]$, as $T\to\infty$. Now, as $S_T$ is Gaussian for $T\in\bN$, we get $\mu_{\gamma}(S_T)=-\mu_T+
\frac{\gamma}{2}\sigma_T^2$ and consequently
\[
\alpha(S_T)=2(m-\lambda)T/\sigma_T^2.
\]
Let us now consider three cases: (1) $\sigma=0$; (2) $\sigma=\infty$; (3) $\sigma\in (0,\infty)$. In the first case, we are in the sub-diffusive regime (since
$\sigma_T$ grows in a sublinear way) and obtain $J_\infty(\phi,\lambda)=\infty$. This is the case for   processes with negative memory, e.g. when $r_t$ is a fractional noise with Hurst parameter $H<1/2$, see e.g. Chapter 1 of \cite{GirKouSur2012}. In the second case, we are in the super-diffusive regime and get $J_\infty(\phi,\lambda)=0$. This corresponds to long memory processes, e.g. a fractional noise with $H>1/2$, see \cite{GirKouSur2012}. Finally, in the third case, we are in the intermediate (diffusive) regime where the
variance of $S_T$ scales linearly in $T$ thus
we obtain a finite, positive $J_{\infty}(\phi,\lambda)$. This is the case with e.g. standard ARMA processes. 
We can clearly see how performance depends on the process memory in this example.
\end{example}

\begin{example}[Long-run performance for a non-logarithmic growth]
The previous two examples show that in certain cases benchmarking with a linear (logarithmic) growth might lead to a degenerate performance. To overcome this problem one can consider a more generic problem with a deterministic benchmark $(G_T)_{T\in\bT}$ given by $G_T:=g(T)$ for some increasing function $g:\bT\to \bR_{+}$ and also consider assessing portfolio value rather than its growth. The objective value of a corresponding long-run portfolio optimization problem would be then given by
\begin{equation}\label{eq:Kg.one}
K(g):=\liminf_{T\to\infty}\sup_{\phi\in\mathcal{A}_T}\alpha\left(V_T(\phi)-g(T)\right),
\end{equation}
where $\mathcal{A}_T$ denotes the family of strategies available for trading up to time $T\in\bT$. This formulation can lead to radically different conclusions. For instance, following the futures market dynamics introduced \cite{GuaNagRas2021}, let us assume that the underlying asset is an Ornstein-Uhlenbeck process, exhibiting a linear mean-reversion, and we can construct a strategy $\phi_T\in\mathcal{A}_T$ for
which $\mu_{\gamma}( V_T(\phi))=cT^2/\gamma$ with some $c>0$. This means, in particular, that for $g_1(T):=\lambda T$, where $\lambda>0$, we always get $K(g_1)=\infty$. In order to have $K(g)<\infty$, one needs at least quadratic $g$. In other words, in such markets, the performance of a well-chosen portfolios can counterbalance a super-linearly growing deterministic return. If the underlying asset exhibits a non-linear mean-reversion, $g$ must be growing as a power higher than $2$ to get finite $K(g)$, see Theorem 2.2 in \cite{GuaNagRas2021}.
\end{example}

\begin{example}[Links to risk-sensitive stochastic control]
For exponential utility, the objective function defined in \eqref{eq:long.run.perf} could be seen as a a map that is dual to the {\it time averaged long-run risk sensitive criterion} objective function given by
\begin{equation}\label{eq:limit.growth}
J^*_{\infty}(\phi,\gamma):=\liminf_{T\to\infty}\frac{-\mu_{\gamma}\left(\ln V_T(\phi)\right)}{T}.
\end{equation}
In \eqref{eq:limit.growth}, given risk aversion $\gamma>0$, we look for an optimal growth rate $\lambda>0$ encoded in strategy that maximises \eqref{eq:limit.growth}, see e.g.~\cite{BieCiaPit2013,BiePli2003} for details. Indeed, since for exponential utility the corresponding certainty equivalent is translation invariant, we get the 
heuristics
\begin{align*}
J_{\infty}(\phi,\lambda) & = \liminf_{T\to\infty}\left[\sup\left\{\gamma>0\colon \mu_{\gamma}\left(\ln V_T(\phi)-\lambda T\right)\leq 0 \right\}\right]\\
& =\liminf_{T\to\infty}\left[\sup\left\{\gamma>0\colon \frac{-\mu_{\gamma}\left(\ln V_T(\phi)\right)}{T}\geq \lambda \right\}\right]\\
&\approx\sup\left\{\gamma>0\colon \liminf_{T\to\infty}\frac{-\mu_{\gamma}\left(\ln V_T(\phi)\right)}{T}\geq \lambda \right\}\\
& =\sup\left\{\gamma>0\colon J^*_{\infty}(\phi,\gamma)\geq \lambda \right\}.
\end{align*}
In particular, while in the classic setting we are looking for a strategy with optimal growth, here we ask ourselves how risk averse we can be to superhedge a prescribed growth. In a stochastic control framework, this might correspond to considering a set of Bellman inequalities instead of a single Bellman equation, see \cite{BiePli1999,Cav2010,BisBor2023}.
\end{example}




\bibliographystyle{siamplain}
\bibliography{bibliography}
\end{document}